\documentclass[11pt, letterpaper]{article}
\usepackage{amsmath, amssymb, amsthm, amsfonts, bbm, commath}
\usepackage[inline]{enumitem}
\usepackage{tikz}
\usetikzlibrary{arrows.meta}
\usepackage{subcaption}
\usepackage{authblk}
\usepackage{graphicx}
\usepackage{hyperref}
\usepackage{xcolor}
\usepackage[ruled]{algorithm2e}

\title{
	Byzantine Consensus under Local Broadcast Model: Tight Sufficient Condition
		\thanks{This research is supported in part by National Science Foundation 
		award 1733872, and Toyota InfoTechnology Center. Any opinions, findings, and conclusions or recommendations expressed here are those of the authors and do not necessarily reflect the views of the funding agencies or the U.S. government.
		}
}

\author[1]{
	Muhammad Samir Khan
}

\author[2]{
	Nitin H. Vaidya
}

\affil[1]{
	Department of Computer Science \protect\linebreak
	University of Illinois at Urbana-Champaign \protect\linebreak
	\texttt{mskhan6@illinois.edu} \protect\linebreak
}

\affil[2]{
	Department of Computer Science \protect\linebreak
	Georgetown University \protect\linebreak
	\texttt{nv198@georgetown.edu}
}

\newtheorem{theorem}{Theorem}[section]
\newtheorem{lemma}[theorem]{Lemma}

\newtheorem{observation}[theorem]{Observation}

\renewenvironment{proof}{\noindent{\bf Proof:} \hspace*{1mm}}{
	\hspace*{\fill} $\Box$ }
\newenvironment{proof_of}[1]{\noindent {\bf Proof of #1:}
	\hspace*{1mm}}{\hspace*{\fill} $\Box$ }

\newcommand{\floor}[1]{\lfloor {#1} \rfloor}

\begin{document}
	\maketitle

	\begin{abstract}
		\normalsize
In this work we consider Byzantine Consensus on undirected communication graphs under the local broadcast model.
In the classical point-to-point communication model the messages exchanged between two nodes $u, v$ on an edge $uv$ of $G$ are private.
This allows a faulty node to send conflicting information to its different neighbours, a property called \emph{equivocation}.
In contrast, in the local broadcast communication model considered here, a message sent by node $u$ is received identically by all of its neighbours.
This restriction to broadcast messages provides non-equivocation even for faulty nodes.
In prior results \cite{NaqviMaster, NaqviBroadcast} it was shown that in the local broadcast model the communication graph must be $(\floor{3f/2}+1)$-connected and have minimum degree at least $2f$ to achieve Byzantine Consensus.
In this work we show that this network condition is tight.
	\end{abstract}

	\section{Introduction} \label{section introduction}
In this work we consider Byzantine Consensus on undirected communication graphs under the local broadcast model.
$n$ nodes with binary input communicate with their neighbours in the communication graph $G$ via message passing to reach consensus on a binary output.
In addition, up to $f$ nodes are byzantine faulty where they may exhibit arbitrary behaviour.
Non-faulty nodes must achieve consensus in finite time and must agree on an input of some non-faulty node.

In the classical point-to-point communication model the messages exchanged between two nodes $u, v$ on an edge $uv$ of $G$ are private.
This allows a faulty node to send conflicting information to its different neighbours, a property called \emph{equivocation}.
In contrast, in the local broadcast communication model considered here, a message sent by node $u$ is received identically by all of its neighbours.
This restriction to broadcast messages provides non-equivocation even for faulty nodes.
This communication model is motivated by wireless broadcast networks where messages sent by an entity are received identically by those in its vicinity.

For point-to-point communication model in undirected graphs a $(2f+1)$-connectivity and $n \ge 3f+1$ are together necessary and sufficient \cite{DOLEV198214, Fischer1986}.
In contrast, prior work \cite{NaqviMaster, NaqviBroadcast} has established that for local broadcast a $(\floor{3f/2}+1)$-connectivity and a minimum degree of $2f$ is necessary.
It was also shown \cite{NaqviMaster, NaqviBroadcast} that $2f$-connectivity is sufficient with local broadcast.
In this work we close this gap by presenting an algorithm for undirected graphs with $(\floor{3f/2}+1)$-connectivity and minimum degree of $2f$.

The rest of the paper is organized as follows.
In Section \ref{section related work} we discuss related work.
Section \ref{section model} formalizes the system model and introduces notation for the setting.
We present our main result in Section \ref{section result}.

	\section{Related Work} \label{section related work}
The Byzantine Consensus problem was introduced by Lamport, Shostak, and Pease \cite{Lamport:1982:BGP:357172.357176, Pease:1980:RAP:322186.322188} who proved that $n \ge 3f+1$ is a tight condition for complete graphs.
Subsequent work \cite{DOLEV198214, Fischer1986} identified necessary and sufficient conditions for arbitrary undirected graphs.
For undirected graphs it is well known that reliable communication between nodes is both necessary and sufficient for consensus.
For directed graphs Tseng and Vaidya \cite{LewisByzantineDirected} gave necessary and sufficient conditions where reliable communication between all nodes is not provided.
All of these works considered point-to-point communication links.

Prior work \cite{Bhandari:2005:RBR:1073814.1073841, Koo:2006:RBR:1146381.1146420, Pagourtzis2017, Pelc:2005:BLB:1062032.1711012, TSENG2015512} has looked at the local broadcast model.
However, they were interested in achieving Byzantine \emph{Broadcast} as opposed to Byzantine \emph{Consensus}, the aim of this paper.
In contrast with the point-to-point communication, the network requirements are different for the two problems with local broadcast.
Zhang and Sundaram \cite{ZhangIterative} gave necessary and sufficient conditions for approximate real valued Byzantine Consensus using an iterative algorithm with local broadcast.
In this work we look at exact binary valued consensus.

The ability of faulty nodes to send conflicting information to its different neighbours is called ``equivocation'' in the literature.
Many works have looked at the effects of restricting equivocation for faulty nodes.
Rabin and Ben-Or \cite{Rabin:1989:VSS:73007.73014} considered a global broadcast model, which is the same as our model on complete graphs, and showed that $n \ge 2f+1$ is both sufficient and necessary for any multiparty protocol using synchronous communication.
Clement et. al. \cite{Clement:2012:PN:2332432.2332490} looked at non-equivocation in complete graphs for asynchronous communication.
Different works \cite{ Fitzi:2000:PCG:335305.335363, Jaffe:2012:PEB:2332432.2332491, Li7516067, Ravikant10.1007/978-3-540-30186-8_32} investigate the impact of limiting equivocation via partial broadcast channels modeled as broadcast over hyperedges.

    \section{System Model and Notation} \label{section model}
We consider $n$ nodes in an undirected communication graph $G = (V, E)$.
Two nodes $u$ and $v$ are \emph{neighbors} if $uv \in E$ is an edge of $G$.
The \emph{degree} of a node $u$ is the number of $u$'s neighbors or, equivalently, the number of edges incident to $u$.
Every node $v \in V$ knows the topology of the communication graph $G$.
The communication in the network is synchronous, allowing nodes to make progress in synchronous rounds.
We consider a \emph{local broadcast} model of communication where a message $m$ sent by $u$ is received identically and correctly by all neighbors of $u$.
In contrast, the classical \emph{point-to-point} model allows a node $u$ to target a specific neighbor for each message sent.
This allows a bad node to possibly sent conflicting information to its different neighbors.
We also assume that when a node receives a message, it can correctly determine the neighbor that transmitted the message.

We consider the \emph{Byzantine} model of faults where a faulty node can exhibit arbitrary behavior, including an adversary with complete knowledge of the current system state to coordinate actions among all faulty nodes.
A maximum of $f$ nodes in $G$ can be faulty.
In the \emph{Byzantine Consensus Problem}, each node has a binary input and the goal is for each node to output a binary value, satisfying the following conditions, in the presence of at most $f$ faulty nodes.
\begin{enumerate}
    [label=\arabic*),topsep=0pt,labelindent=0pt]
    \item \textbf{Agreement:}
    All non-faulty nodes must output the same value.
    \item \textbf{Validity:}
    The output of each non-faulty node must be an input of some non-faulty node.
    \item \textbf{Termination:}
    All non-faulty nodes must decide on their output in finite time.
\end{enumerate}

For two nodes $u$ and $v$, a $uv$-path $P_{uv}$ is a path from $u$ to $v$.
$u$ is the \emph{source} node of $P_{uv}$ and $v$ is the \emph{terminal} node of $P_{uv}$.
Every other node in $P_{uv}$ is an \emph{internal} node of $P_{uv}$.
$uv$-paths are \emph{node disjoint} if they do not have any common internal nodes.
For a set $U$ and node $v \not \in U$, a $Uv$-path is a $uv$-path for some node $u \in U$.
$Uv$-paths are \emph{node disjoint} if they do not have any common source or internal nodes.
A $uv$-path or a $Uv$-path is said to \emph{exclude} a set of nodes $X$ if no internal node of the path belongs to $X$.
Note that the source and terminal nodes can still be from $X$.
A $uv$-path or a $Uv$-path is \emph{fault-free} if it does not have any faulty internal node, i.e. excludes the set of faulty nodes.
Note that we allow a fault-free path to have a faulty source and/or terminal node.

A graph $G$ is $k$-connected if $\abs{V(G)} > k$ and removal of any $\le k$ nodes does not disconnect $G$.
By Menger's Theorem \cite{west2001introduction} a graph $G$ is $k$-connected if and only if for any two nodes $u, v \in V$ there exist $k$ node disjoint $uv$-paths.
Another standard result \cite{west2001introduction} for $k$-connected graphs is that if $G$ is $k$-connected, then for any node $v$ and a set of at least $k$ nodes $U$ there exist $k$ node disjoint $Uv$-paths.

	\section{Our Result} \label{section result}
Our main result is that if $G$ is $(\floor{3f/2} + 1)$-connected and each node has degree at least $2f$, then this is sufficient for Byzantine Consensus under the local broadcast model.

\begin{theorem} \label{theorem sufficiency}
    Under the local broadcast model, Byzantine Consensus with at most $f$ faults is achievable if $G$ is $(\floor{3f/2}+1)$-connected and every node in $G$ has degree at least $2f$. 
\end{theorem}

In contrast the point-to-point communication model requires that $G$ must be $(2f+1)$-connected and must have at least $3f+1$ nodes \cite{DOLEV198214, Fischer1986}.
Prior work \cite{NaqviMaster, NaqviBroadcast} has identified the graph condition in Theorem \ref{theorem sufficiency} to be necessary under local broadcast, so the result in this paper implies that the condition is tight.
We prove Theorem \ref{theorem sufficiency} constructively by providing an algorithm that solves Byzantine Consensus under these assumptions and an accompanying proof of correctness.
	    
	    \subsection{Algorithm}
In this section we describe an algorithm to solve consensus on a graph $G$ such that it is $(\floor{ 3f / 2 } + 1)$-connected and each node has degree at least $2f$.
Every communication in the algorithm is a message $m = (L, b, P)$ where $L$ is a label, $b \in \set{ 0, 1 }$ is a binary valued message body, and $P$ is a path where the source node indicates the node where the message originated and the remaining nodes indicate how the message has travelled so far.

A node $v$ communicates with the rest of the graph via ``flooding'' where $v$ sends a message to the entire network with a unique label corresponding to each iteration of flooding.
When a node $u$ receives a message $m = (L, b, P)$ from one of its neighbors $w$ such that $w \not \in P$, then $u$ forwards $m' = (L, b, P + w)$ to its neighbors, if $P + w$ does not contain $n-1$ nodes already.
Since communication is synchronous, $v$ can wait for $n$ synchronous rounds to ensure that the message has propagated to the entire network.
If a faulty node $u$ chooses to stay silent when a message is expected, then the non-faulty neighbors of $u$ replace silence with a default message so that we can assume that a message is always sent when expected, even though it may be tampered.
Similarly if a faulty node $u$ sends a message with a label $L'$ when a label $L$ is expected, then the non-faulty neighbors of $u$ replace the label $L'$ in the message with $L$ when forwarding the message.
If a faulty node $u$ sends a message of the form $(L, b, P)$ where path $P + u$ does not exist, then its non-faulty neighbors ignore it, so that we can assume that messages are always sent along actual paths that exist.
To achieve this, all nodes know the topology of the graph $G$.

In the local broadcast model we can ensure that even a faulty node can only flood a single value with any given label, which is in contrast with the point-to-point communication model where a faulty node can send different conflicting values to its neighbors.
If a faulty node attempts to flood two different values with a given label, then its non-faulty neighbors simply chose to forward first such value, ignoring the rest.
In general, for any label $L$ and path $P$, a non-faulty node will forward only the first message of the form $(L, b, P)$ it receives from a neighbor $u$, ignoring the rest.

\begin{algorithm}[H] \label{algorithm consensus}
    \SetAlgoLined
    \SetKwFor{For}{For}{do}{end}
    Each node $v$ has binary input $\mathtt{input}_v \in \set{ 0, 1 }$ and maintains a binary state $\gamma_v \in \set{ 0, 1 }$.\\
    Initially, $\gamma_v := \mathtt{input}_v$.\\
    \For{every \emph{candidate faulty} set $F \subseteq V$ such that $\abs{F} \le f$}
    {\begin{enumerate}
        [label=Step (\alph*),topsep=0pt,labelindent=12pt,leftmargin=!]
        \item
        Flood $\gamma_v$ with label $F$.
        In the remaining steps of the iteration, values received\\ with label $F$ are considered implicitly.
        
        \item
        For each node $u \in V$, identify a single $uv$-path $P_{uv}$ that excludes $F$.
        Set
        \begin{align*}
            Z_v &:= \set{ u \in V \mid \text{$v$ received $0$ from $u$ along $P_{uv}$ in step (a)} },\\
            N_v &:= \set{ u \in V \mid \text{$v$ received $1$ from $u$ along $P_{uv}$ in step (a)} }.
        \end{align*}
        
        \item
        Define $A_v$ and $B_v$ as follows.
        \begin{enumerate}
            [label=Case \arabic*:,topsep=0pt,labelindent=0pt]
            \item $\abs{Z_v \cap F} \le \floor{f/2}$ and $\abs{N_v} > f$.
            Set $A_v := N_v$ and $B_v := Z_v$.
            
            \item $\abs{Z_v \cap F} \le \floor{f/2}$ and $\abs{N_v} \le f$.
            Set $A_v := Z_v$ and $B_v := N_v$.
            
            \item $\abs{Z_v \cap F} > \floor{f/2}$ and $\abs{Z_v} > f$.
            Set $A_v := Z_v$ and $B_v := N_v$.
            
            \item $\abs{Z_v \cap F} > \floor{f/2}$ and $\abs{Z_v} \le f$.
            Set $A_v := N_v$ and $B_v := Z_v$.
        \end{enumerate}
        If $v \in A_v$, then $\gamma_v$ stays unchanged.
        If $v \in B_v$, then identify $f+1$ node disjoint $A_v v$-paths that exclude $F$.
        If $v$ received identically $0$ (resp. $1$) along these $f+1$\\ paths in step (a), then set $\gamma_v$ to be $0$ (resp. $1$); otherwise $\gamma_v$ stays unchanged.
    \end{enumerate}}
    
    Output $\gamma_v$.
    
    \caption{Algorithm for Byzantine Consensus under local broadcast model. All steps are performed by a node $v$.}
\end{algorithm}

The formal procedure is presented as Algorithm \ref{algorithm consensus} and we give a formal proof of correctness in Section \ref{section correctness}.
The algorithm is inspired by the Byzantine Consensus algorithm for directed graphs by Tseng and Vaidya \cite{LewisByzantineDirected}.
Here we describe the algorithm informally.
Each node $v$ starts with a binary input and maintains a binary state $\gamma_v$ which is initially set to its input.
The algorithm has one loop and in each iteration we select a set $F$ such that $\abs{F} \le f$.
We call $F$ the \emph{candidate faulty} set of the iteration.
As a first step, each node $v$ floods its state $\gamma_v$.
Note that this requires $v$ to 1) locally broadcast $\gamma_v$ to its neighbors and 2) assist other nodes to flood their states by forwarding messages received from its neighbors, as described earlier.

Let $Z$ be the set of nodes that flooded $0$ in the first step and let $N$ be the set of nodes that flooded $1$.
Note that $Z$ and $N$ partition the node set $V$.
In step (b), each non-faulty node $v$ finds its own estimate $Z_v$ and $N_v$ of $Z$ and $V$ respectively, using $F$ as the candidate faulty set.
This is done by selecting, for each node $u$, a single path $P_{uv}$ that excludes $F$.
Such a path must exist since $G$ is $(\floor{3f/2}+1)$-connected.
Observe that if $u = v$, then $P_{vv}$ is the single node path consisting of exactly $v$.
If $v$ receives $0$ along $P_{uv}$, then $v$ puts $u$ in $Z_v$ and if $v$ receives $1$ along this path, then $v$ puts $u$ in $N_v$.
Note that if there exists a faulty node in $V - F$, then $Z_v$ and $N_v$ may not be the same as $Z$ and $N$.
It is also possible that, for a node $w \ne v$, $Z_v \ne Z_w$ and $N_v \ne N_w$.
However, $Z_v$ and $N_v$ partition the node set $V$.

In step (c), using $F$ as the candidate faulty set, either nodes in $Z$ change their state to $1$ or nodes in $N$ change their state to $0$.
Since each non-faulty node $v$'s estimates $Z_v$ and $N_v$ may vary when faulty nodes exist in $V - F$, non-faulty nodes may disagree on which nodes must retain their state and which nodes must switch.
The four cases in Algorithm \ref{algorithm consensus} step (c) select sets $A_v$ and $B_v$ so that, from $v$'s perspective, nodes in $A_v$ retain their state and nodes in $B_v$ switch.
$A_v$ and $B_v$ are selected so that each node $u$ in $B_v$ has at least $f+1$ node disjoint $A_v u$-paths excluding $F$ (Lemma \ref{lemma propagates}).
This ensures that 1) when $V - F$ contains faulty nodes, then non-faulty nodes do not switch to a state of a faulty node (Lemma \ref{lemma validity}) and 2) when $F$ contains all the faulty nodes, then non-faulty nodes in $B_v$ do switch their states correctly to achieve agreement (Lemma \ref{lemma agreement}).

At the end, each node $v$ outputs its state $\gamma_v$ as its decision.
	        
        \subsection{Proof of Correctness} \label{section correctness}
We assume $f > 0$ in the proof.
Theorem \ref{theorem sufficiency} is trivially true for $f=0$.
At a high level, our objective is two folds.
Firstly, we want to show that regardless of the choice of the candidate faulty set $F$, a non-faulty node $v$ will only change its state in step (c) of the iteration to match a state of some non-faulty node at the beginning of the iteration (Lemma \ref{lemma validity}).
By a simple induction, it follows that the final state of $v$ is an input of some non-faulty node.
Secondly, we want to show that when $F$ contains the actual set of faulty nodes in a given execution, then all non-faulty nodes reach agreement in that iteration (Lemma \ref{lemma agreement}).
Since all non-faulty nodes now have the same state, by Lemma \ref{lemma validity} it follows that in the future iterations all non-faulty nodes will maintain this state.

We start by observing that if a message is received along a fault-free path, then the exact same message was indeed sent.
Recall that even a faulty node can only flood a single value in step (a) of the algorithm.

\begin{observation} \label{observation fault-free reliable}
    For any iteration of the loop in Algorithm \ref{algorithm consensus}, in step (a) for any two nodes $u, v \in V$ (possibly faulty), $v$ receives $b$ along a fault-free $uv$-path $P_{uv}$ if and only if $u$ flooded $b$.
\end{observation}

A priori, it is not clear if the paths identified in the algorithm always exist.
For paths in step (b) the existence is relatively straightforward to prove.

\begin{lemma}
    For any two nodes $u, v$ and any iteration of the loop in Algorithm \ref{algorithm consensus} with a candidate faulty set $F$, there exists a $uv$-path that excludes $F$.
\end{lemma}

\begin{proof}
    Let $F' = F - u - v$.
    Since $\abs{F'} \le f$ and $G$ is $(\floor{3f/2}+1)$-connected, we have that $G - F'$ is a connected graph and so there exists a $uv$-path in $G - F'$ which is a $uv$-path in $G$ that excludes $F$.
\end{proof}

The next lemma shows that the choice of sets $A_v$ and $B_v$ ensures that paths in step (c) also exist.

\begin{lemma} \label{lemma propagates}
    For any non-faulty node $v$ and any iteration of the loop in Algorithm \ref{algorithm consensus} with a candidate faulty set $F$, in step (c) if $v \in B_v$, then there exist $f+1$ node disjoint $A_v v$-paths that exclude $F$.
\end{lemma}
\begin{proof}
    Fix an iteration in the algorithm and the candidate faulty set $F$.
    Consider an arbitrary non-faulty node $v$ such that $v \in B_v$ in step (c).
    There are 4 cases to consider, corresponding to the 4 cases in step (c).
    \begin{enumerate}
        [label=Case \arabic*:,topsep=0pt,labelindent=0pt]
        \item $\abs{Z_v \cap F} \le \floor{f/2}$ and $\abs{N_v} > f$.
        Then $A_v := N_v$ and $B_v := Z_v$.
        Therefore there exist at least $f+1$ nodes in $A_v$.
        Node $v$ selects $f+1$ nodes $A'_v$ from $A_v$ by choosing all nodes from $A_v \cap F$ and the rest arbitrarily from $A_v - F$.
        By choice of $B_v$, we have that $B'_v = B_v \cap (F-v)$ has at most $\floor{f/2}$ nodes.
        Since $G$ is $(\floor{3f/2}+1)$-connected, so there exist $f+1$ node disjoint $A'_v v$-paths in $G$ that exclude $B'_v$.
        Furthermore, since all the nodes in $A_v \cap F$ are the source nodes in these paths and $F = (A_v \cap F) \cup (B_v \cap F)$, we have that these paths exclude $F$\footnote{recall that a path that excludes $X$ does not have nodes from $X$ as internal nodes; however, nodes from $X$ may be the source or terminal nodes.}.
        
        \item $\abs{Z_v \cap F} \le \floor{f/2}$ and $\abs{N_v} \le f$.
        Then $A_v := Z_v$ and $B_v := N_v$.
        Since the degree of $v$ is at least $2f$ and there are at most $f$ nodes in $B_v$ (including $v$), we have that $v$ has at least $f+1$ neighbors in $A_v$.
        There are therefore $f+1$ node disjoint $A_v v$-paths that have no internal nodes and hence exclude $F$.
        
        \item $\abs{Z_v \cap F} > \floor{f/2}$ and $\abs{Z_v} > f$.
        Then $A_v := Z_v$ and $B_v := N_v$.
        We have that
        \begin{align*}
            \abs{N_v \cap F}
                &=      \abs{F} - \abs{Z_v \cap F}  \\
                &\le    f - \floor{f/2} - 1 \\
                &\le    \floor{f/2}.
        \end{align*}
        So this case is the same as Case 1 with the roles of $Z_v$ and $N_v$ swapped.
        
        \item $\abs{Z_v \cap F} > \floor{f/2}$ and $\abs{Z_v} \le f$.
        Then $A_v := N_v$ and $B_v := Z_v$.
        From the analysis in Case 3, we have that $\abs{N_v \cap F} \le \floor{f/2}$.
        So this case is the same as Case 2 with the roles of $Z_v$ and $N_v$ swapped.
    \end{enumerate}
    In all four cases we have that there do exist $f+1$ node disjoint $A_v v$-paths that exclude $F$.
\end{proof}

We now show that in any iteration, the state of a non-faulty node at the end of the iteration equals the state of some non-faulty node at the beginning of the iteration. 

\begin{lemma} \label{lemma validity}
    For any non-faulty node $v$ and any iteration of the loop in Algorithm \ref{algorithm consensus}, the state of $v$ at the end of the iteration is equal to a state of some non-faulty node $u$ at the beginning of the iteration.
\end{lemma}
\begin{proof}
    Fix an iteration in the algorithm and the candidate faulty set $F$.
    For any node $u$, we denote the state at the beginning of the iteration by $\gamma_u^{\operatorname{start}}$ and the state at the end of the iteration by $\gamma_u^{\operatorname{end}}$.
    Consider an arbitrary non-faulty node $v$ and the sets $A_v$ and $B_v$ in step (c).
    Let $P_1, \dots, P_{f+1}$ be the $f+1$ node disjoint $A_v v$-paths identified by $v$ in step (c).
    If $\gamma_v^{\operatorname{start}} = \gamma_v^{\operatorname{end}}$, then the claim is trivially true.
    So suppose $v \in B_v$ and $v$ receives identical values along $P_1, \dots, P_{f+1}$ in step (a).
    Since the number of faulty nodes is at most $f$, thus at least one of these paths is both fault-free and has a non-faulty source node, say $u$.
    By Observation \ref{observation fault-free reliable}, it follows that whatever value is received by $v$ along this path in step (a) is the value flooded by $u$.
    Therefore $\gamma_v^{\operatorname{end}} = \gamma_u^{\operatorname{start}}$, where $u$ is a non-faulty node, as required.
\end{proof}

Next, we show that when the candidate faulty set $F$ is properly selected, all non-faulty nodes reach agreement in that iteration.

\begin{lemma} \label{lemma agreement}
    Consider an iteration of the loop in Algorithm \ref{algorithm consensus} such that all faulty nodes are contained in the candidate faulty set $F$.
    Then for any two non-faulty nodes $u, v \in V$, we have that $u$ and $v$ have the same state at the end of the iteration.
\end{lemma}
\begin{proof}
    Fix an iteration of the algorithm and the candidate faulty set $F$ such that all faulty nodes are contained in $F$.
    Let $Z$ be the set of nodes that flooded $0$ in step (a) of the iteration and let $N$ be the set of nodes that flooded $1$ in step (a).
    We first show that for any non-faulty node $v$, $Z_v = Z$ and $N_v = N$.
    Consider an arbitrary node $w$ that flooded $0$ (resp. $1$) in step (a) of the iteration so that $w \in Z$ (resp. $w \in N$).
    Observe that $P_{wv}$ identified in step (b) of the iteration excludes $F$ and is fault-free.
    Therefore, by Observation \ref{observation fault-free reliable} $v$ receives $0$ (resp. $1$) along $P_{wv}$ and correctly sets $w \in Z_v$ (resp. $w \in N_v$), as required.
    
    It follows that for any two non-faulty nodes $u$ and $v$, we have that $Z_u = Z_v = Z$ and $N_u = N_v = N$.
    Thus $A_u = A_v$ and $B_u = B_v$.
    Let $A := A_u$ and $B := B_u$.
    Now all nodes in $A$ flooded identical value in step (a), say $\alpha$.
    If $u \in A$, then $u$'s state is $\alpha$ at the beginning of the iteration and stays unchanged in step (c).
    Therefore, at the end of the iteration $\gamma_u = \alpha$.
    If $u \in B$, then observe that the $f+1$ node disjoint $A u$-paths identified by $u$ in step (c) are all fault-free.
    By Observation \ref{observation fault-free reliable}, it follows that $u$ receives $\alpha$ identically along these $f+1$ paths and so, at the end of the iteration, $\gamma_u = \alpha$.
    Similarly for $v$, we have that $\gamma_v = \alpha$, as required.
\end{proof}

We now have all the necessary ingredients to prove Theorem \ref{theorem sufficiency}.

\begin{proof_of}{Theorem \ref{theorem sufficiency}}
    The algorithm terminates in finite time at every node.
    This satisfies the termination property.
    Validity follows from Lemma \ref{lemma validity} via a simple induction.
    For agreement, observe that there exists at least one iteration of the loop in Algorithm \ref{algorithm consensus} where all faulty nodes are contained in the candidate faulty set $F$, namely when $F$ is exactly the set of faulty nodes.
    From Lemma \ref{lemma agreement} we have that all non-faulty nodes have the same state at the end of this iteration.
    From Lemma \ref{lemma validity}, it follows that non-faulty nodes do not change their states in subsequent iterations.
    Therefore, all non-faulty nodes output the same state at the end of the algorithm.
\end{proof_of}
    
    \section{Discussion} \label{section discussion}
In this work, we have presented a constructive proof that it is sufficient for the communication graph $G$ to be $(\floor{3f/2}+1)$-connected and have degree at least $2f$ for Byzantine Consensus under local broadcast.
The results in this paper along with prior results in \cite{NaqviMaster, NaqviBroadcast} show that the condition that $G$ must be $(\floor{3f/2}+1)$-connected and have degree at least $2f$ is tight.
The algorithm provided in this paper is not polynomial.
In \cite{NaqviMaster, NaqviBroadcast} we gave an efficient algorithm when $G$ is $2f$-connected.
We leave finding an efficient algorithm for the tight condition for future work. 

	\bibliographystyle{abbrv}
	\bibliography{bib}

\end{document}